\documentclass[aps,prd,twocolumn,preprintnumbers]{revtex4-2}
\usepackage{amssymb,amsmath,amsfonts}
\usepackage{epsf,epsfig}
\usepackage[utf8]{inputenc}
\usepackage{color}
\usepackage{slashed}
\usepackage{tensor} 
\usepackage{braket}
\usepackage{amsthm}
\usepackage{enumitem}
\usepackage{float}
\usepackage[normalem]{ulem} 
\usepackage[colorlinks=true]{hyperref}  
\hypersetup{
    bookmarks=true,         
    unicode=false,          
    pdftoolbar=true,        
    pdfmenubar=true,        
    pdffitwindow=false,     
    pdfstartview={FitH},    
    colorlinks=true,       
    linkcolor=magenta, 
    citecolor=blue,        
    filecolor=magenta,      
    urlcolor=cyan           
} 

\newcommand{\bea}{\begin{eqnarray}}

\newcommand{\eea}{\end{eqnarray}}
\usepackage{lipsum} 
\newcommand{\ba}{\begin{eqnarray}}
\newcommand{\ea}{\end{eqnarray}}

\newcommand{\beq}{\begin{equation}}
\newcommand{\eeq}{\end{equation}}
\newcommand{\beqa}{\begin{eqnarray}}
\newcommand{\eeqa}{\end{eqnarray}}
\newcommand{\beqar}{\begin{eqnarray*}}
\newcommand{\eeqar}{\end{eqnarray*}}

\usepackage{color}

\newtheorem{theorem}{Theorem}

\theoremstyle{definition}
\newtheorem{example}{Example}

\begin{document}

\title{Novel duality-invariant theories of electrodynamics}
\author{\'Angel J. Murcia}
\email{angelmurcia@icc.ub.edu}
\affiliation{Departament de F\'isica Qu\`antica i Astrof\'isica, Institut de Ci\`encies del Cosmos\\
 Universitat de Barcelona, Mart\'i i Franqu\`es 1, E-08028 Barcelona, Spain }

\begin{abstract}
We identify new families of duality-invariant theories of electrodynamics. We achieve this in two different ways. On the one hand, we present an algorithm to construct a one-parameter family of exactly duality-invariant theories from a single seed duality-invariant theory. If the seed theory is causal, the theories constructed from this method will also be causal when the parameter is non-negative. On the other hand, we find two additional novel families of duality-invariant theories which include a nonzero term independent of the electromagnetic field. The first of them generalizes Bialynicki-Birula electrodynamics, while the second family of theories features a well-defined Maxwell limit as the term independent of the gauge field strength is sent to zero. 
\end{abstract} 

\maketitle


{\bf Introduction.} Symmetries arguably conform the most powerful tool towards the construction of physically meaningful theories. They provide the fundamental building blocks with which to write the subsequent Lagrangians. This is remarkably illustrative in electrodynamics: demanding Lorentz and gauge invariance, one finds that a theory for a vector field $A_a$ must be entirely expressed in terms of Lorentz scalars built from the gauge-invariant field strength $F_{ab}=2\partial_{[a} A_{b]}$. The simplest theory of electrodynamics corresponds to Maxwell theory, which has been extremely successful in the description of electromagnetic interactions. 

But one could consider other models of electrodynamics fulfilling Lorentz and gauge invariance beyond Maxwell theory. These are called theories of non-linear electrodynamics (NED) and have been extensively studied for a long time. Some prominent models include Born-Infeld theory \cite{Born:1934gh} --- introduced to provide a finite self-energy for the electric field of a charged point particle ---  or the Euler-Heisenberg Lagrangian \cite{Heisenberg:1936nmg}, incorporating one-loop quantum corrections --- for other NED models and additional background on the topic, we refer the reader to, e.g., \cite{Boillat:1970gw,Plebanski:1970zz,Bialynicki-Birula:1992rcm,Gibbons:1995cv,Gibbons:1995ap,Gaillard:1997zr,Gaillard:1997rt,Dunne:2012vv,Gaete:2014nda,kruglov2015model,Kruglov:2016ezw,Bandos:2020jsw,Bandos:2020hgy,Gullu:2020ant,Tahamtan:2020lvq,Avetisyan:2021heg,Sorokin:2021tge,Babaei-Aghbolagh:2022uij,Ferko:2022iru,Mkrtchyan:2022ulc,Lechner:2022qhb,Babaei-Aghbolagh:2022itg,Russo:2022qvz,Ferko:2023ruw,Russo:2024ptw,Russo:2024kto,Russo:2025fuc,Chen:2025ndc}. Nonetheless, if one really takes symmetries as the true guiding principle for the construction of physical theories, nothing prevents us from considering other theories of NED which do not violate other fundamental principles of physics (such as the absence of ghosts or causality).

In fact, Maxwell theory is not only Lorentz- and gauge-invariant: it is also invariant under electromagnetic duality rotations. As a result, the literature has extensively studied duality-invariant NEDs, the most renowned examples being Born-Infeld theory and the celebrated ModMax electrodynamics \cite{Bandos:2020jsw}. However, as the duality-invariance condition adopts the form of a quadratic partial differential equation, the finding of further explicit solutions\footnote{The general solution to this equation is known, although in an implicit way \cite{Gaillard:1997rt}.} has been a quite elusive task, with some notable exceptions \cite{Gibbons:1995cv,Hatsuda:1999ys}.

The present manuscript contributes in this direction in a two-fold way. On one side, we identify a novel solution-generating technique that allows one to construct a one-parameter family of duality-invariant NEDs from a seed duality-invariant one. This algorithm relies crucially on the properties of the ModMax Lagrangian and preserves the causality of the seed theory --- in case it features this property --- when the parameter associated with the construction is nonnegative. On the other side, we also find two novel instances of duality-invariant theories featuring an additional term $\psi$ independent of $F_{ab}$ with intriguing properties: the first conforms the most general fixed point of the solution-generating method described above, while the second family of theories has a well-defined Maxwell limit as $\psi \rightarrow 0$. 

\vspace{0.15cm}

{\bf Setup.} Consider a generic NED $\mathcal{L}=\mathcal{L}(F_{ab}, \star F_{ab})$ on four-dimensional flat space-time $\eta_{ab}$, where $(\star F)_{ab}=\frac{1}{2} \varepsilon_{abcd} F^{cd}$ and $\varepsilon_{abcd}$ denotes the components of the totally antisymmetric tensor (with $\varepsilon_{0123}=1$). As it is well known, the most general such a theory may be expressed as a function of the invariants:
\begin{equation}
\label{eq:defsp}
s=-\frac{1}{4} F_{ab} F^{ab}\,, \quad p=-\frac{1}{4} F_{ab} (\star F)^{ab}\,.
\end{equation}
In terms of these variables, the condition for $\mathcal{L}=\mathcal{L}(s,p)$ to be duality-invariant was found in \cite{Bialynicki-Birula:1992rcm,Gibbons:1995cv}. Since it will be important for us, let us briefly sketch its derivation. Define the two-form field strength $H_{ab}$ as follows:
\begin{equation}
\label{eq:defH}
H_{ab}=2 \star \frac{\partial \mathcal{L}}{\partial F^{ab}}\,.
\end{equation}
In terms of $H$, the (generalized) Maxwell equation and Bianchi identity read $\mathrm{d} H=0$ and $\mathrm{d} F=0$. Consider a electromagnetic duality rotation into new fields $(F',H')$:
\begin{equation}
\begin{pmatrix}
F' \\ H'
\end{pmatrix}=\begin{pmatrix}
\cos \alpha  & \sin \alpha \\ -\sin \alpha & \cos \alpha
\end{pmatrix}\begin{pmatrix}
F \\ H
\end{pmatrix}\,.
\label{eq:dualrot}
\end{equation}
Clearly, $\mathrm{d} H'=\mathrm{d} F'=0$. However, to enjoy duality-invariance, we must demand the invariance of \eqref{eq:defH}. In particular, if we rewrite \eqref{eq:defH} as $H=\mathfrak{f}(F)$, we must require that $H'=\mathfrak{f}(F')$. Infinitesimally, this translates into:
\begin{equation}
\frac{1}{2}\star F_{ab}= \frac{\partial}{\partial F^{ab}} \left ( \frac{\partial \mathcal{L}}{\partial F_{cd}} \star \frac{\partial \mathcal{L}}{\partial F^{cd}} \right) \,.
\end{equation}
This equation may be readily integrated to yield:
\begin{equation}
\frac{1}{4} F_{ab} \star F^{ab}=  \frac{\partial \mathcal{L}}{\partial F_{ab}} \star \frac{\partial \mathcal{L}}{\partial F^{ab}}-\psi\,,
\end{equation}
where $\psi$ is independent of $F_{ab}$. This term is commonly set to zero to ensure a Maxwell weak-field limit \cite{Gibbons:1995cv}. However, we will not set it systematically to zero, as it will be justified afterwards. This term $\psi$ may be equivalently obtained from the Hamiltonian formalism --- we refer the interested reader to the Supplemental Material. 

Using the invariants $(s,p)$ defined in \eqref{eq:defsp}, one may rewrite the condition for duality invariance as:
\begin{equation}
p \left (\frac{\partial \mathcal{L}}{\partial s} \right)^2-2 s \frac{\partial \mathcal{L}}{\partial s}\frac{\partial \mathcal{L}}{\partial p}- p \left (\frac{\partial \mathcal{L}}{\partial p} \right)^2+\psi=p\,.
\label{eq:spdual}
\end{equation}
There are several well-known instances of solutions to \eqref{eq:spdual} in the literature in the case $\psi=0$. The simplest example is provided by Maxwell theory $\mathcal{L}=s$. Other prominent instances are given by Born-Infeld electrodynamics \cite{Born:1934gh}:
\begin{equation}
\mathcal{L}^{\mathrm{BI}}=T- \sqrt{T^2-2 T s-p^2}\,, \quad T \in \mathbb{R}\,,
\label{eq:bilag}
\end{equation}
or the celebrated ModMax electrodynamics \cite{Bandos:2020jsw}:
\begin{equation}
\mathcal{L}_\gamma^{\mathrm{MM}}=\cosh \gamma \, s+ \sinh \gamma \, \sqrt{s^2+p^2}\,, \quad \gamma \in \mathbb{R}\,,
\label{eq:modmax}
\end{equation}
which represents the most general theory of non-linear electrodynamics which is also conformally invariant.

It turns out that there is a method to generate an infinite family of duality-invariant theories of non-linear electrodynamics from a given solution of \eqref{eq:spdual}. Indeed, let $\mathcal{L}_0(s,p)$ an arbitrary such solution. Define:
\begin{equation}
\mathcal{M}_\gamma(\mathcal{L}_0)(s,p)=\mathcal{L}_0(\mathcal{L}_\gamma^{\mathrm{MM}}(s,p),p)\,, \quad \gamma \in \mathbb{R}\,.
\label{eq:modmaxmap}
\end{equation}
Observe that $\mathcal{M}_\gamma(\mathcal{L}_0)$ is obtained by replacing every appearance of $s$ in $\mathcal{L}_0$ by $\mathcal{L}_\gamma^{\mathrm{MM}}$. The map $\mathcal{M}_\gamma$ picks a given theory and generates a one-parameter family of NEDs. We will refer to $\mathcal{M}_\gamma$ as the \emph{ModMax map}. As the following theorem shows, if the ModMax map is applied on a duality-invariant theory, the resulting family of theories for any $\gamma$ is also duality-invariant.

\begin{theorem}
\label{thm}
Let $\mathcal{L}_0(s,p)$ be a given duality-invariant theory. Then the one-parameter family of theories $\mathcal{M}_\gamma(\mathcal{L}_0)(s,p)=\mathcal{L}_0(\mathcal{L}_\gamma^{\mathrm{MM}}(s,p),p)$ 
 is duality-invariant.
\end{theorem}
\begin{proof}
If $z=\mathcal{L}_\gamma^{\mathrm{MM}}$ and $\mathcal{M}_\gamma(\mathcal{L}_0)=\hat{\mathcal{L}}_\gamma$:
\begin{equation}
\frac{\partial \hat{\mathcal{L}}_\gamma}{\partial s}=\frac{\partial \mathcal{L}_0}{\partial z} \frac{\partial z}{\partial s}\,,\quad 
\frac{\partial \hat{\mathcal{L}}_\gamma}{\partial p}=\frac{\partial \mathcal{L}_0}{\partial z} \frac{\partial z}{\partial p}+\frac{\partial \mathcal{L}_0}{\partial p}\,,
\end{equation}
where $\mathcal{L}_0$ is viewed as a function of $(z,p)$. Direct computation shows that:
\begin{align}
\nonumber
&p \left (\frac{\partial \hat{\mathcal{L}}_\gamma}{\partial s} \right)^2-2 s \frac{\partial \hat{\mathcal{L}}_\gamma}{\partial s}\frac{\partial \hat{\mathcal{L}}_\gamma}{\partial p}- p \left (\frac{\partial \hat{\mathcal{L}}_\gamma}{\partial p} \right)^2=\\ \nonumber
& \left ( \frac{\partial \mathcal{L}_0}{\partial z} \right)^2 \left (p \left ( \frac{\partial z}{\partial s} \right)^2-2 s \frac{\partial z}{\partial s}\frac{\partial z}{\partial p}- p \left (\frac{\partial z}{\partial p} \right)^2 \right)\\&-2 \frac{\partial \mathcal{L}_0}{\partial z}
 \frac{\partial \mathcal{L}_0}{\partial p} \left (s \frac{\partial z}{\partial s}+p\frac{\partial z}{\partial p}\right)-p \left (\frac{\partial \mathcal{L}_0}{\partial p} \right)^2\,.
\end{align}
Using that $z=\mathcal{L}_\gamma^{\mathrm{MM}}$ is conformally and duality-invariant, so that it satisfies \eqref{eq:spdual} with $\psi=0$ and 
\begin{equation}
s \frac{\partial z}{\partial s}+p \frac{\partial z}{\partial p}= z\,,
\end{equation}
one directly learns that $\hat{\mathcal{L}}_\gamma$ satisfies \eqref{eq:spdual} with the same $\psi$ as $\mathcal{L}_0$ and we conclude. 
\end{proof}
Theorem \ref{thm} allows to construct families of exactly duality-invariant theories from a given duality-invariant theory, generalizing some known perturbative results in the literature in the context of $T\bar{T}$-flows \cite{Babaei-Aghbolagh:2024uqp}.

The iterative application of the ModMax map does not produce further families of duality-invariant theories. This can be seen from the following property of the ModMax Lagrangian \eqref{eq:modmax}:
\begin{equation}
\mathcal{M}_\gamma(\mathcal{L}_{\kappa}^{\mathrm{MM}} )=\mathcal{L}_{\gamma+\kappa}^{\mathrm{MM}}\,.
\end{equation}
Let us now present some non-trivial families of theories obtained from the application of the ModMax map on some well-known theories of electrodynamics.

\begin{example}
Under \eqref{eq:modmaxmap}, the Born-Infeld Lagrangian produces:
\begin{equation}
\mathcal{M}_\gamma(\mathcal{L}^{\mathrm{BI}})=T- \sqrt{T^2-2 T \mathcal{L}_\gamma^{\mathrm{MM}}-p^2}\,.
\label{eq:genbi}
\end{equation}
This theory corresponds to generalized Born-Infeld theory \cite{Bandos:2020jsw,Bandos:2020hgy}. From this new perspective, duality invariance of \eqref{eq:genbi} is a direct consequence of Theorem \ref{thm}.
\end{example}

\begin{example}
\label{ep:gr1}
Take the duality-invariant theories:
\begin{align}
& \hspace{1cm}\mathcal{L}^{\mathrm{GR}}=\sqrt{\alpha\mp  2u} \sqrt{\kappa \pm 2v}\,,\\
\label{eq:uvdef}
&2u=\sqrt{s^2+p^2}-s\,, \quad 2v=\sqrt{s^2+p^2}+s\,.
\end{align} 
where $\{\alpha, \kappa\}$  are constants. Applying the ModMax map:
\begin{equation}
\mathcal{M}_\gamma(\mathcal{L}^{\mathrm{GR}})= \sqrt{\alpha\mp 2e^{-\gamma} u } \sqrt{\kappa \pm 2e^{\gamma} v}\,.
\end{equation}
This family coincides\footnote{\emph{En passant} we fix a small typo present in Equation (4.1) of \cite{Gibbons:1995cv}. According to their notation, they should have $\beta \delta=-1$.} with the class of duality-invariant theories found by Gibbons and Rasheed in \cite{Gibbons:1995cv}. 
\end{example}

\begin{example}
Take the following duality-invariant theory considered by Hatsuda, Kamimura and Sekiya \cite{Hatsuda:1999ys}:
\begin{equation}
\mathcal{L}^{\mathrm{HKS}}=\frac{(\theta+bs)\sqrt{1-(\theta-b s)^2}-\arcsin\left[\theta-bs \right]}{2b}\,,
\end{equation}
with $\theta=\sqrt{b^2 s^2+2b\sqrt{s^2+p^2}}$ and $b \in \mathbb{R}$. Applying \eqref{eq:modmaxmap}:
\begin{align}
\nonumber
\mathcal{M}_\gamma(\mathcal{L}^{\mathrm{HKS}})&=\frac{(\theta_\gamma+b\mathcal{L}_\gamma^{\mathrm{MM}})\sqrt{1-(\theta_\gamma-b \mathcal{L}_\gamma^{\mathrm{MM}})^2}}{2b}\\&-\frac{\arcsin\left[\theta_\gamma-b\mathcal{L}_\gamma^{\mathrm{MM}} \right]}{2b}\,,
\end{align}
where $\theta_\gamma=\sqrt{b^2 \left (\mathcal{L}_\gamma^{\mathrm{MM}} \right)^2+2b\sqrt{\left (\mathcal{L}_\gamma^{\mathrm{MM}} \right)^2+p^2}}$. This new family of theories is manifestly different from the original one $\mathcal{L}^{\mathrm{HKS}}$, since their weak-field expansion reads\footnote{Note that $\mathcal{O}\left ( s^2,p^2 \right)$ stands for terms that, under a rescaling $(s,p) \rightarrow (as, ap)$, get an overall factor $a^{k}$ with $k \geq 2$.}:
\begin{align}
\nonumber
\mathcal{M}_\gamma(\mathcal{L}^{\mathrm{HKS}})=\mathcal{L}_\gamma^{\mathrm{MM}}&- \frac{2 \sqrt{2 b}}{3} \left(\left (\mathcal{L}_\gamma^{\mathrm{MM}} \right)^2+ p^2\right)^{3/4} \\& +\mathcal{O}\left (s^{2},p^{2} \right )\,.
\end{align}
\end{example}

\begin{example}
\label{ep:gr2}
Take the family of duality-invariant theories obtained in \cite{Gibbons:1995cv} and given by\footnote{Observe that the theory is not well defined for $p=0$.}:
\begin{align}
\label{eq:gr2}
&\hspace{-0.2cm}\mathcal{L}_{\pm}^{\mathrm{GR},\psi}=\mathfrak{f}_K (\sqrt{u}+\sqrt{v},0)\pm\mathfrak{f}_K (\sqrt{u}-\sqrt{v},\mathrm{sign}(p)\, \psi ), \hspace{-0.1cm} \\ \nonumber
&\hspace{-0.2cm}\mathfrak{f}_K(b,c)=\frac{b}{2} \sqrt{2(K^2-c)-b^2}\\ &\hspace{1.5cm}+(K^2-c) \arcsin\left ( \frac{b}{\sqrt{2(K^2-c)}}\right)\,,
\end{align}
where $(b,c)$ are some ancillary variables to define $\mathfrak{f}_K$,  $K \in \mathbb{R}$ and $(u,v)$ were defined in \eqref{eq:uvdef}. The theory \eqref{eq:gr2} solves \eqref{eq:spdual} for arbitrary $\psi$ and for both sign choices $\pm$. Applying \eqref{eq:modmaxmap}, one gets a novel family of duality-invariant theories given by:
\begin{align}
\label{eq:gr2mm}
\mathcal{M}_\gamma \left (\mathcal{L}_{\pm}^{\mathrm{GR},\psi} \right)&=\mathfrak{f}_K(e^{-\gamma/2}\sqrt{u}+e^{\gamma/2}\sqrt{v},0)\\&\nonumber \pm\mathfrak{f}_K(e^{-\gamma/2}\sqrt{u}-e^{\gamma/2}\sqrt{v},{\rm sign}(p)\,\psi)\,.
\end{align}
\end{example}

{\bf Causality and ModMax map}. Consider a duality-invariant theory $\mathcal{L}$ satisfying \eqref{eq:spdual} with $\psi=0$. In terms of the variables $(u,v)$ introduced in \eqref{eq:uvdef}, Eq.\ \eqref{eq:spdual} reads: 
\begin{equation}
\frac{\partial \mathcal{L}}{\partial u} \frac{\partial \mathcal{L}}{\partial v}=-1\,. 
\label{eq:spdualuv}
\end{equation}
As explained in \cite{courant1962methods}, this equation may be implicitly solved as follows:
\begin{equation}
\mathcal{L}[\ell(\tau)]=\ell(\tau)-\frac{2 u}{\dot{\ell}(\tau)}\,, \quad \tau=v+\frac{u}{\dot{\ell}(\tau)^2}\,.
\label{eq:paramch}
\end{equation}
where $\ell(\tau)$ is an arbitrary function of $\tau \geq 0$ and $\dot{\ell}=\mathrm{d}\ell/d\tau$. Indeed, if the Lagrangian may be expressed as \eqref{eq:paramch}, then
\begin{equation}
\mathrm{d} \mathcal{L}=- \mathrm{d}u/\dot{\ell}+ \dot{\ell} \mathrm{d}v\,,
\end{equation}
so that \eqref{eq:spdualuv} is trivially satisfied. Let us explore how the ModMax map acts when the original theory is expressed implicitly as in \eqref{eq:paramch}. In terms of the $(u,v)$ variables, one has that
\begin{equation}
\mathcal{M}_\gamma (\mathcal{L})(u,v)=\mathcal{L}(e^{-\gamma}u, e^{\gamma}v)\,.
\end{equation}
Since $\mathcal{M}_\gamma (\mathcal{L})$ is duality-invariant by Theorem \ref{thm}, the parametrization \eqref{eq:paramch} applies and:
\begin{equation}
\mathcal{M}_\gamma (\mathcal{L})[q(\sigma)]=q(\sigma)-\frac{2 u}{\dot{q}(\sigma)}\,, \quad \sigma=v+\frac{u}{\dot{q}(\sigma)^2}\,,
\label{eq:parmm}
\end{equation}
for a certain function $q$. Now, note that:
\begin{align}
\mathrm{d}\left ( \mathcal{M}_\gamma (\mathcal{L})\right) &=-e^{-\gamma}/\dot{\ell}\, \mathrm{d}u+e^{\gamma} \dot{\ell} \mathrm{d}v=-\mathrm{d}u/\dot{q}+\dot{q} \, \mathrm{d}v\,,
\end{align}
where $\dot{\ell}$ is evaluated at $\tau(e^{-\gamma}u, e^{\gamma}v)$ and $\dot{q}$ at $\sigma(u,v)$. As a consequence, we conclude that
\begin{equation}
\dot{q}(\sigma(u,v))=e^{\gamma}\dot{\ell}(\tau(e^{-\gamma}u, e^{\gamma}v))\,.
\label{eq:edoparmm}
\end{equation}
Using this result in \eqref{eq:paramch}, we infer that $\tau(e^{-\gamma}u, e^{\gamma}v))=e^{\gamma} \sigma(u,v)$. Substituting in \eqref{eq:edoparmm}, we learn that\footnote{Strictly, the conclusion would be $q(\sigma)=\ell(e^{\gamma} \sigma)+ b$, for $b \in \mathbb{R}$. However, this just amounts to the addition of an inoffensive constant to the Lagrangian, so we simply disregard it.}
\begin{equation}
q(\sigma)=\ell(e^{\gamma} \sigma)\,.
\label{eq:modmaxpar}
\end{equation}
Therefore, in terms of the implicit representation \eqref{eq:paramch}, the ModMax map amounts to replacing $\ell(\tau)$ by $\ell(e^{\gamma} \tau)$, which is entirely consistent with existing examples in the literature, such as in generalized Born-Infeld \cite{Russo:2024llm}.

This result has strong implications for the subsequent family of duality-invariant theories obtained by the application of the ModMax map on a seed theory $\mathcal{L}$. Indeed, parametrizing $\mathcal{L}$ as in \eqref{eq:paramch}, it was shown in \cite{Russo:2024ptw} that $\mathcal{L}$ is causal if and only if:
\begin{equation}
\dot{\ell} \geq 1\,, \quad \ddot{\ell} \geq 0\,.
\label{eq:cauconds}
\end{equation}
Equation \eqref{eq:modmaxpar} implies the following result.
\begin{theorem}
\label{thm2}
Let $\mathcal{L}$ be a given duality-invariant causal theory. Take $\gamma \geq 0$. Then $\mathcal{M}_\gamma(\mathcal{L})$ is causal too.
\end{theorem}
\begin{proof}
As proven above, $\mathcal{M}_\gamma(\mathcal{L})$ can be expressed implicitly as in \eqref{eq:parmm} with $q(\sigma)=\ell(e^{\gamma} \sigma)$, cf. Eq.\ \eqref{eq:modmaxpar}. Hence
\begin{equation}
\dot{q}(\sigma)=e^{\gamma} \dot{\ell}(e^{\gamma} \sigma)\,, \quad \ddot{q}(\sigma)=e^{2\gamma} \ddot{\ell}(e^{\gamma} \sigma)\,.
\end{equation}
Therefore, if $\gamma \geq 0$ and the seed theory satisfies the causality constraints \eqref{eq:cauconds}, the new family of duality-invariant theories $\mathcal{M}_\gamma(\mathcal{L})$ will satisfy $\dot{q} \geq 1$ and $\ddot{q} \geq 0$ and we conclude. 
\end{proof}
Note that we can refine a bit more Theorem \ref{thm2}. Indeed, assume a theory $\mathcal{L}$ that is expressed as in \eqref{eq:paramch} and consider that $\dot{\ell} \geq \dot{\ell}_0$ for a real number $\dot{\ell}_0>0$. Then, $\mathcal{M}_\gamma(\mathcal{L})$ with $\gamma \geq -\log(\dot{\ell}_0)$ will be causal. If $\dot{\ell}_0 >1$, negative values of $\gamma$ are allowed. If $\dot{\ell}_0 <1$ --- what implies that the starting theory was not causal ---, the ModMax map allows one to construct causal generalizations of the seed theory for $\gamma \geq -\log(\dot{\ell}_0)$. 

\vspace{0.15cm}

{\bf Novel duality-invariant theories with $\psi \neq 0$}. Instances of duality-invariant theories with non-trivial $\psi$ were already known --- just consider Example \ref{ep:gr2}, extracted from \cite{Gibbons:1995cv}. It is the purpose of this section to present novel duality-invariant theories with nonzero $\psi$. The justification for this is two-fold: there are popular NEDs with no Maxwell limit for weak fields (e.g., ModMax) and some NEDs with $\psi \neq 0$ have a Maxwell limit as $\psi \rightarrow 0$ --- see theory \eqref{eq:ame} below. 

Let us briefly comment about the possible interpretation of $\psi$. As commented above, the term $\psi$ in the duality-invariance condition \eqref{eq:spdual} is only required to be independent of $F_{ab}$. Therefore, the most simple choice for $\psi$ is to take it as a pure integration constant --- just a real number ---, with the same units as $s$ and $p$.

However, there are other possibilities \emph{a priori}. In fact, $\psi$ could be built from other fields with non-trivial dynamics, as long as these are independent of $F_{ab}$. Nevertheless, promoting $\psi$ to be dynamical will generically bring about an important hindrance: the equations of motion for the fields composing $\psi$ will \emph{not} generically be invariant under duality rotations --- so it is not clear if one could even speak of duality invariance in these cases. To see this, take a theory $\mathcal{L}=\mathcal{L}(s,p,\psi)$ satisfying \eqref{eq:spdual}. Under rotations \eqref{eq:dualrot} of small angle $\alpha$ between $F_{ab}$ and $H_{ab}$:
\begin{equation}
\delta_{\alpha} \left( \frac{\partial \mathcal{L}}{\partial \psi} \right)=\alpha\,,
\label{eq:trafoderpsi}
\end{equation}
so that the equations for the additional fields conforming $\psi$ will not transform in general covariantly. A possible resolution of this problem would be to impose that $\psi$ is given by a total derivative term. In such a case, the relevant dangerous terms $\dfrac{\partial \mathcal{L}}{\partial \psi}$ will always be affected by derivatives in the equations for the fields composing $\psi$. As a result, \eqref{eq:trafoderpsi} will guarantee the invariance of the equations. Various possible choices for $\psi$ that correspond to total derivatives
and provide equations for the fields composing $\psi$ that are invariant under duality transformations are presented in the Supplemental Material.


Let us now study solutions of \eqref{eq:spdual} with non-trivial $\psi$. Its direct resolution conforms a daunting task, so it is convenient to look for special solutions with additional properties. Particularly, we will restrict to theories $\mathcal{L}=\mathcal{L}(s,p,\psi)$ which are homogeneous functions of degree one, i.e., $\mathcal{L}(\Theta s,\Theta p, \Theta \psi)=\Theta \mathcal{L}$ for any function $\Theta$, so that
\begin{equation}
s \frac{\partial \mathcal{L}}{\partial s}+p \frac{\partial \mathcal{L}}{\partial p}+ \psi \frac{\partial \mathcal{L}}{\partial \psi}=\mathcal{L}\,.
\label{eq:hdifcond}
\end{equation}

We will be interested in finding NEDs that satisfy the duality-invariance condition \eqref{eq:spdual} \emph{and} \eqref{eq:hdifcond}. Clearly, this way we will not be able to find the most general solution of \eqref{eq:spdual}, but instead we will find valuable instances of novel duality-invariant theories with a non-trivial $\psi$. Amusingly, note that \eqref{eq:hdifcond} would represent the condition for conformal invariance of the theory if $\psi$ transforms under a conformal transformation $\eta_{ab} \rightarrow \Omega^2 \eta_{ab}$ as $\psi \rightarrow \Omega^{-4} \psi$. 

In first place, one could resort to a perturbative approach and identify instances of duality-invariant NEDs up to a certain perturbative order in $\psi$ which additionally fulfill \eqref{eq:hdifcond}. We show in the Supplemental Material that this approach is indeed useful, describing the algorithm to produce duality-invariant theories up to any finite order in $\psi$, as well as the explicit form of the most general solution to \eqref{eq:spdual} and \eqref{eq:hdifcond} up to sixth order in $\psi$.

But it is possible to go beyond a perturbative analysis and find exact solutions of \eqref{eq:spdual} with non-trivial $\psi$ and satisfying \eqref{eq:hdifcond}. A first family of solutions, as well as its one-parameter generalization provided by the ModMax map, was provided in Example \ref{ep:gr2}. Next we will present two additional novel families of theories fulfilling \eqref{eq:spdual} and \eqref{eq:hdifcond}:  the first example will conform the most general duality-invariant theory that remains invariant under the action of the ModMax map \eqref{eq:modmaxmap}, while the second family of theories represents, to the best of our knowledge, the first theories with non-trivial $\psi$ that have a well-defined Maxwell limit as $\psi \rightarrow 0$.

\begin{example}
Let us look for duality-invariant theories which do not depend on one of the invariants $s$ or $p$. The duality invariance condition \eqref{eq:spdual} with non-trivial $\psi$ is inconsistent with a theory independent of $p$, so we may only hope to find theories which are independent of $s$. In such a case, Eq. \ \eqref{eq:spdual} dramatically simplifies into
\begin{equation}
-p \left ( \frac{\partial \mathcal{L}}{\partial p} \right)^2+\psi=p\,,
\end{equation}
whose most general solution is:
\begin{equation}
\mathcal{L}_{\rm g BB}=\sqrt{(\psi-p)p}+\psi \arctan \left [ \sqrt{\frac{p}{\psi-p}}\right]+\mathfrak{h}(\psi) \,,
\label{eq:genbb}
\end{equation}
where $\mathfrak{h}$ is an arbitrary function of $\psi$ and where we require $p(\psi-p) \geq 0$ for reality of $\mathcal{L}_{\rm g BB}$. Observe that \eqref{eq:hdifcond} is satisfied automatically if  $\mathfrak{h}$ is a linear function.  Up to a trivial overall sign, the theory \eqref{eq:genbb} represents the \emph{most general} duality-invariant NED which is independent of $s$. As a result, it conforms the most general fixed point of the ModMax map (preserving duality invariance). It is to be interpreted as a generalization of Bialynicki-Birula electrodynamics \cite{Bialynicki-Birula:1992rcm}, whose Lagrangian $\mathcal{L}=i \vert p \vert$ is obtained as $\psi \rightarrow 0$ in \eqref{eq:genbb} (it corresponds to a topological term). Clearly, \eqref{eq:genbb} does not have a Maxwell limit. 

Define $y=\dfrac{p}{\psi}$. In the weak-field limit $1 >>y>0$ and disregarding the term $\mathfrak{h}(\psi)$, $\mathcal{L}_{\rm g BB}$ reduces to
\begin{align}
&\mathcal{L}_{\rm g BB}=2\psi \sqrt{y} \left[ 1- \frac{y}{6} -\frac{y^2}{40 } +\mathcal{O}\left ( y^3 \right)\right]\,, \quad \psi > 0\,,\\
&\mathcal{L}_{\rm g BB}=\frac{2\psi y^{3/2}}{3} \left[ 1 +\frac{3y}{10 } +\frac{9y^2}{56}+\mathcal{O}\left ( y^3 \right)\right]\,, \quad \psi < 0\,.
\end{align}

\end{example}
\begin{example}
Consider now the following theory:
\begin{align}
\label{eq:ame}
&\hspace{-0.2cm}\mathcal{L}_{\psi{\rm M}}(s,p,\psi)=\frac{s}{\sqrt{2}}\mathfrak{t}_++\frac{p}{\sqrt{2}} \mathfrak{t}_--\frac{\psi}{2}\mathfrak{t}_\psi\,,\\
&\hspace{-0.2cm}\mathfrak{t}_+=\sqrt{ 1- \frac{\psi\, p}{s^2+p^2}+\sqrt{\frac{s^2+(\psi-p)^2}{s^2+p^2}}}\,, \\
&\hspace{-0.2cm}\mathfrak{t}_-=\frac{1}{\psi s}\sqrt{\psi^2 s^2\left[ \frac{ \psi p}{s^2+p^2}-1+\sqrt{\frac{s^2+(\psi-p)^2}{s^2+p^2}} \right]}\,,\\
&\hspace{-0.2cm}\mathfrak{t}_\psi=\arctan\left(\frac{s}{p}\right)+2\arctan\left (\frac{ \mathfrak{t}_-}{\sqrt{2}+\mathfrak{t_+}} \right)\,.
\end{align}
The theory $\mathcal{L}_{\psi{\rm M}}$ may be seen to satisfy \eqref{eq:spdual} for nonzero $\psi$ after noting the following identities:
\begin{align}
\frac{\partial \mathcal{L}_{\psi{\rm M}}}{\partial s}&=\frac{\mathfrak{t}_+}{\sqrt{2}}\,, \quad \frac{\partial\mathcal{L}_{\psi{\rm M}}}{\partial p}=\frac{\mathfrak{t}_-}{\sqrt{2}}\,, \quad \frac{\partial \mathcal{L}_{\psi{\rm M}}}{\partial \psi}=-\frac{\mathfrak{t}_\psi}{2}\,, \\
\mathfrak{t}_+^2&=\mathfrak{t}_-^2+2-\frac{2 \psi \,p}{s^2+p^2}\,, \quad \mathfrak{t}_+ \mathfrak{t}_-=\frac{\psi \,s}{s^2+p^2}\,.
\end{align}
Consequently, Eq.\ \eqref{eq:ame} represents an instance of a novel duality-invariant theory of electrodynamics.  As $\psi \rightarrow 0$, the theory \eqref{eq:ame} reduces naturally to Maxwell theory:
\begin{equation}
\lim_{\psi \rightarrow 0} \mathcal{L}_{\psi{\rm M}}=s\,. 
\end{equation}
In addition, $\mathcal{L}_{\psi{\rm M}}$ is always real for all real values of $s,p$ and $\psi$.  This can be seen from the fact that:
\begin{align}
\frac{s^2+(\psi-p)^2}{s^2+p^2}=\frac{\psi^2 s^2}{(s^2+p^2)^2}+\left (1-\frac{\psi p}{s^2+p^2} \right)^2\,.
\end{align}
In another vein, observe that $\mathcal{L}_{\psi{\rm M}}$ is not analytic for null fields $s=p=0$. Nevertheless, this is already a well-known feature of some celebrated duality-invariant theories, including ModMax electrodynamics \cite{Sorokin:2021tge}.

In the weak-field limit $\psi^2 > >s^2+p^2>0$, we get the following approximation\footnote{Again, by $\mathcal{O}\left (x^{3/2},y^{3/2} \right)$ we denote terms that, under a rescaling $(x,y) \rightarrow (ax, ay)$, get an overall factor $a^k$ with $k \geq 3/2$.} for $\mathcal{L}_{\psi{\rm M}}$:
\begin{align}
\nonumber
\mathcal{L}_{\psi{\rm M}}&=-\psi\left [ \frac{\arctan \left(\frac{x}{y} \right)}{2}+\,\mathrm{sign}(x) \arctan \frac{\omega^+ }{\omega^-} \right]\\& \hspace{-1cm}+\frac{\psi \left [ (y +\sqrt{x^2+y^2}) \omega^+ + \vert x \vert \omega^-   \right]}{\sqrt{2}(x^2+y^2)^{3/4} \, \mathrm{sign}(x)}+\mathcal{O}\left (x^{3/2},y^{3/2} \right)\,,
\end{align}
where $\psi y =p$ as above and:
\begin{align}
 \omega^{\pm}=\sqrt{x^2+y^2 \pm y \sqrt{x^2+y^2}}\,, \quad 
x=\frac{s}{\psi}\,.
\end{align}
We explicitly see that no Maxwell limit is obtained in this regime. One may derive an additional one-parameter family of duality-invariant theories of electrodynamics from \eqref{eq:ame} by direct application of the ModMax map \eqref{eq:modmaxmap}:
\begin{equation}
\mathcal{M}_\gamma (\mathcal{L}_{\psi{\rm M}})=\mathcal{L}_{\psi{\rm M}}(\mathcal{L}_\gamma^{\rm MM}(s,p),p,\psi)\,.
\label{eq:gammame}
\end{equation}
Note that $\lim_{\psi \rightarrow 0} \mathcal{M}_\gamma (\mathcal{L}_{\psi{\rm M}})=\mathcal{L}_\gamma^{\rm MM}$. 
\end{example}

{\bf Discussion}. In this work we have found various novel duality-invariant theories of electrodynamics. We have achieved this through two different approaches. On the one hand, we have identified a general method to construct one-parameter families of duality-invariant theories starting from any given one. Strongly rooted in the properties of the ModMax Lagrangian, we have called such an algorithm the \emph{ModMax map} and obtained new families of duality-invariant theories by direct application of this map to various well-known examples in the literature. For those theories satisfying \eqref{eq:spdual} with $\psi=0$, we have also showed that the ModMax map also respects the causality properties of the starting theory if the parameter introduced by the ModMax map is nonnegative.

On the other hand, we have also identified novel duality-invariant theories that satisfy \eqref{eq:spdual} for non-trivial $\psi$. One of these new theories represents a generalization of Bialynicki-Birula electrodynamics, while the second family of theories features a well-defined Maxwell limit as $\psi \rightarrow 0$ --- correspondingly, a ModMax limit if one considers the one-parameter generalization constructed from the ModMax map. We were able to find these new solutions to \eqref{eq:spdual} by requiring them to be a homogeneous function of degree one in the variables $s$, $p$ and $\psi$, which imposes the differential constraint \eqref{eq:hdifcond}. Interestingly enough, this condition is tantamount to the conformal invariance of the subsequent Lagrangian if $\psi$ transforms appropriately under conformal transformations. However, if $\psi$ is assumed to be composed by dynamical fields, the corresponding equations for the fields making up $\psi$ would not be covariant under duality rotations, thus hampering the duality invariance of the theory --- unless for very special choices for $\psi$, such as taking it to be a total derivative. In any event, one may regard the additional constraint \eqref{eq:hdifcond} as a computational trick to identify novel duality-invariant theories for non-trivial $\psi$ --- conceiving it is a pure integration constant if necessary.

The present manuscript poses the scrutiny of various research directions in the near future. On the one hand, it would be interesting to examine which physical properties of a given duality-invariant theory (beyond causality) will be preserved by the ModMax map if $\gamma \geq 0$. Among these potential properties, birefringence \cite{Boillat:1970gw,Plebanski:1970zz,Bialynicki-Birula:1992rcm} will not be one of them: while Born-Infeld is non-birefringent, its ModMax generalization does feature birefringence \cite{Russo:2022qvz}. 

On the other hand, the study of duality-invariant theories with non-trivial $\psi$ represents an intriguing avenue. One could explore the Hamiltonian formalism of such theories, analyzing whether a simple map relating the Lagrangian and Hamiltonian pictures of duality-invariant theories may exist, just like in theories of pure electrodynamics \cite{Russo:2024ptw,Russo:2025fuc}. Also, among the novel theories with non-zero $\psi$ that we have identified, the theory \eqref{eq:ame} conforms a particularly interesting one, as it possesses a well-defined Maxwell limit. It is natural to wonder what further unique physical properties may possess this theory. More generally, it is tantalizing to investigate whether there are other duality-invariant theories having Maxwell theory as weak-field limit and which are analytic for all possible values of the invariants $s$ and $p$, including null fields. What physical properties would these theories have? This work only represents the first step towards the exploration of these questions.

{\bf Acknowledgments}. \'AJM would like to thank Jorge Russo, Dmitri Sorokin and Paul Townsend for enlightening conversations. \'AJM was supported by a Juan de la Cierva contract (JDC2023-050770-I) from Spain’s Ministry of Science, Innovation and Universities.

\onecolumngrid  
\begin{center}  
{\Large\bf Appendices} 
\end{center} 
\appendix 
\allowdisplaybreaks

\section{Duality-invariance condition with $\psi \neq 0$ in the Hamiltonian formalism}\label{appA}

In this section, we will briefly describe how a nonzero $\psi$ appears in the condition that requires duality invariance in the Hamiltonian picture --- we refer the reader to \cite{Birula:1975book,Bialynicki-Birula:1992rcm} for useful background on the topic. In this formalism, the fundamental fields correspond to the electric induction $\mathbf{D}$ and the magnetic induction $\mathbf{B}$. If $\mathcal{H}=\mathcal{H}(\mathbf{D},\mathbf{B})$ denotes the Hamiltonian density, the corresponding generalized Maxwell equations of the theory are given by:
\begin{align}
\dot{\mathbf{B}}&=-\mathbf{\nabla} \times \mathbf{E}\,, \quad \mathbf{\nabla} \cdot \mathbf{B}=0\,, \\
\dot{\mathbf{D}}&=\mathbf{\nabla} \times \mathbf{H}\,, \quad \mathbf{\nabla} \cdot \mathbf{D}=0\,, 
\end{align}
where the electric and magnetic intensities $\mathbf{E}$ and $\mathbf{H}$ are defined as:
\begin{equation}
\mathbf{E}=\frac{\partial \mathcal{H}}{\partial \mathbf{D} }\,, \quad \mathbf{H}=\frac{\partial \mathcal{H}}{\partial \mathbf{B} }\,.
\end{equation}
Accordingly, if $\mathcal{L}$ stands for the Lagrangian of the theory, we have:
\begin{equation}
\mathbf{D}=\frac{\partial \mathcal{L}}{\partial \mathbf{E} }\,, \quad \quad \mathbf{H}=-\frac{\partial \mathcal{L}}{\partial \mathbf{B} }\,.
\end{equation}
Duality rotations act by shifting the phase of the complex vector $\mathbf{D}+i \mathbf{B}$. As in the Lagrangian formalism, we only require duality invariance to be a symmetry of the equations of motion, not of the theory itself. Infinitesimal duality rotations act on $\mathbf{D}$ and $\mathbf{B}$ as $\delta \mathbf{D}=-\alpha \mathbf{B}$ and  $\delta \mathbf{B}=\alpha \mathbf{D}$ for infinitesimal parameter $\alpha$. If the Hamiltonian density transforms as $\delta \mathcal{H}=\alpha \psi$, where $\psi$ is independent of $\mathbf{D}$ and $\mathbf{B}$, then the theory will be duality-invariant, as the set of (electromagnetic) equations of motion will remain unchanged. In particular:
\begin{equation}
\delta \mathcal{H}=\frac{\partial \mathcal{H}}{\partial \mathbf{D} } \cdot \delta \mathbf{D}+\frac{\partial \mathcal{H}}{\partial \mathbf{B} } \cdot \delta \mathbf{B}=-\alpha \mathbf{E} \cdot \mathbf{B}+\alpha \mathbf{D}\cdot \mathbf{H}=\alpha \psi\,.
\end{equation}
As a result, we end up with:
\begin{equation}
 \mathbf{D} \cdot \mathbf{H}=\mathbf{E}\cdot \mathbf{B}+\psi\,.
\end{equation}
Observe that this exact equation was found in \cite{Gibbons:1995cv} --- see their equation (2.7). Now, if the starting Lagrangian $\mathcal{L}$ is gauge and Lorentz invariant, so that $\mathcal{L}=\mathcal{L}(s,p)$:
\begin{equation}
\mathbf{D}=\frac{\partial \mathcal{L}}{\partial \mathbf{E} }=\frac{\partial \mathcal{L}}{\partial s}\mathbf{E}+\frac{\partial \mathcal{L}}{\partial p}\mathbf{B}\,, \quad \mathbf{H}=-\frac{\partial \mathcal{L}}{\partial \mathbf{B} }=\frac{\partial \mathcal{L}}{\partial s}\mathbf{B}-\frac{\partial \mathcal{L}}{\partial p}\mathbf{E}\,,
\end{equation}
by using that $s=\frac{1}{2}\left ( \vert \mathbf{E} \vert^2- \vert \mathbf{B} \vert^2 \right)$ and $p= \mathbf{E} \cdot \mathbf{B}$, we end up precisely with \eqref{eq:spdual}. As a consequence, we conclude that the appearance of a nonzero $\psi$ is due to the fact that the Hamiltonian \emph{per se} is not invariant under duality rotations. Nevertheless, this is not a worrisome feature, as it is only the set of equations of motion the one that needs to remain invariant under duality rotations, exactly as in the Lagrangian picture. This is the reason why $\psi$ does not appear in relevant works such as \cite{Bandos:2020jsw} --- see their equation (8) ---, where they explicitly look for Hamiltonians which remain themselves invariant under duality rotations.

Interestingly enough, this implies that the Hamiltonian density will not be solely a function of the rotational invariants $ \vert \mathbf{D} \vert^2+ \vert \mathbf{B} \vert^2$ and $\vert \mathbf{D} \times \mathbf{B} \vert$. Indeed, it will also include a functional dependent on $\mathbf{D} \cdot \mathbf{B}$, which is not invariant under duality rotations --- note that the Hamiltonians will be Lorentz-invariant by construction, as they are related to Lorentz-invariant Lagrangians.

\section{Instances of dynamical choices for $\psi$}\label{appB}

Let us consider a certain theory satisfying the duality-invariance condition \eqref{eq:spdual}. As explained in the body of the manuscript, allowing $\psi$ to be constructed from other dynamical fields entails an important danger: the equations of motion for such additional fields might not be invariant under duality transformations. One may solve this problem if $\psi$ corresponds to a total derivative term, as it is illustrated below with three different examples.


\begin{enumerate}
\item $\psi=\lambda \Box \Phi$, where $\Phi$ is a scalar field, $\Box$ stands for the Laplacian operator and $\lambda$ a certain coupling with units of length$^{-1}$. Adding the scalar kinetic term, one could consider the scalar-vector theory:
\begin{equation}
\tilde{\mathcal{L}}_{F_{ab},\Phi}=\frac{\beta}{2}\Phi \Box \Phi+\mathcal{L}(s,p,\lambda \Box \Phi)\,,
\label{eq:escalarvector}
\end{equation}
where $\beta$ is a dimensionless coupling and  $\mathcal{L}$ satisfies \eqref{eq:spdual} with $\psi=\lambda \Box \Phi$. The equation of motion for the scalar reads:
\begin{equation}
 \beta \Box  \Phi+ \lambda \Box \left(  \frac{\partial \mathcal{L}}{\partial \psi}\right)=0\,.
 \label{eq:escalarveom}
\end{equation}
Since $\frac{\partial \mathcal{L}}{\partial \psi}$ is necessarily affected by derivatives in \eqref{eq:escalarveom}, \eqref{eq:trafoderpsi} ensures that \eqref{eq:escalarveom} remains indeed invariant under electromagnetic duality rotations. Nonetheless, the equations will feature higher derivatives.

\item $\psi=\lambda \mathfrak{p}$, where $\lambda$ is a dimensionless coupling and $\mathfrak{p}=-\frac{1}{4}\mathcal{F}_{ab} \left (\star \mathcal{F} \right) ^{ab}$ for a field strength $\mathcal{F}_{ab}=2\partial_{[a} \mathcal{A}_{b]}$ of a different gauge field $\mathcal{A}_a$. Consider now theories of the form:
\begin{equation}
\tilde{\mathcal{L}}_{F_{ab} \mathcal{F}_{ab}}=\hat{\mathcal{L}}(\mathfrak{s},\mathfrak{p})+\mathcal{L}(s,p,\lambda \mathfrak{p})\,,
\label{eq:exdoblecampo}
\end{equation} 
where $\mathfrak{s}=-\frac{1}{4}\mathcal{F}_{ab}\mathcal{F}^{ab}$, $\hat{\mathcal{L}}$ is an arbitrary function of the indicated variables and where $\mathcal{L}(s,p,\psi)$ satisfies \eqref{eq:spdual}. The equation of motion for $\mathcal{F}_{ab}$ is given by:
\begin{equation}
\partial_a \left (  \frac{\partial \hat{\mathcal{L}}}{\partial \mathfrak{s}}  \mathcal{F}^{ab}\right) +   \partial_a \left ( \frac{\partial \hat{\mathcal{L}}}{\partial \mathfrak{p}} +\lambda\, \frac{\partial \mathcal{L}}{\partial \psi}  \right)  \star \mathcal{F}^{ab}=0\,.
\label{eq:pextra}
\end{equation}
Again, we directly see that \eqref{eq:pextra} is exactly invariant under electromagnetic duality rotations. One could additionally wonder whether there exist theories $\hat{\mathcal{L}}(\mathfrak{s},\mathfrak{p})$ for which the complete theory $\tilde{\mathcal{L}}_{F_{ab} \mathcal{F}_{ab}}$ is \emph{also} invariant under  duality rotations associated with the extra gauge field $\mathcal{A}_a$. This lies beyond the scope of the present manuscript and is planned to be studied elsewhere. 
\item $\psi=\lambda \chi_{\rm GB}$, where $\lambda$ is a dimensionless coupling and $\chi_{\rm GB}=R_{abcd} R^{abcd}-4 R_{ab}R^{ab}+R^2$ stands for the Gauss-Bonnet density. Take the following theory:
\begin{equation}
\tilde{\mathcal{L}}_{F_{ab}, g_{ab}}= \frac{R}{16 \pi G_N}+\mathcal{L}(s,p,\lambda \chi_{\rm GB})\,,
\end{equation}
where $\mathcal{L}(s,p,\psi)$ fulfills \eqref{eq:spdual}. Define the following tensor \cite{Padmanabhan:2011ex} $P_{ab}{}^{cd}$:
\begin{equation}
P_{ab}{}^{cd}=\frac{\partial \tilde{\mathcal{L}}_{F_{ef}, g_{ef}}}{\partial R_{cd}{}^{ab}}=\frac{1}{16\pi G_N} \delta_{[a}^{c} \delta_{b]}^{d}+2\lambda \frac{\partial \mathcal{L}}{\partial \psi} \Sigma_{ab}{}^{cd} \,, \quad \Sigma_{ab}{}^{cd}=R_{ab}{}^{cd}-4 R_{[a}^{[c} \delta_{b]}^{d]}+R \delta_{[a}^{c} \delta_{b]}^{d}\,.
\end{equation}
The gravitational equations of motion for the theory $\tilde{\mathcal{L}}_{F_{ab}, g_{ab}}$ will be given by \cite{Cano:2020qhy}:
\begin{equation}
P_{(a}{}^{cde} R_{b)cde}-\frac{1}{2}g_{ab}\tilde{\mathcal{L}}_{F_{ef}, g_{ef}}+2 \nabla_c \nabla_d P_{(a}{}^c{}_{b)}{}^d+\frac{\partial \mathcal{L}}{\partial F^{(a}{}_c } F_{b)c}=0\,.
\label{eq:geneomgrav}
\end{equation}
We compute:
\begin{align}
P_{(a}{}^{cde} R_{b)cde}=\frac{R_{ab}}{16 \pi G_N} +\frac{\lambda}{2} \frac{\partial \mathcal{L}}{\partial \psi} \chi_{\rm GB} g_{ab}\,, \quad
2 \nabla_c \nabla_d P_{(a}{}^c{}_{b)}{}^d=4\lambda \nabla_c \nabla_d \left ( \frac{\partial \mathcal{L}}{\partial \psi}\right) \Sigma_{a}{}^c{}_b{}^d\,, \quad   \,.
\end{align}
As a result, \eqref{eq:geneomgrav} simplifies down to:
\begin{equation}
\frac{G_{ab}}{16 \pi G_N}+\frac{1}{2} g_{ab} \left[ \lambda \frac{\partial \mathcal{L}}{\partial \psi} \chi_{\rm GB}- \mathcal{L} \right]+4\lambda \nabla_c \nabla_d \left ( \frac{\partial \mathcal{L}}{\partial \psi}\right)  \Sigma_{a}{}^c{}_b{}^d+\frac{\partial \mathcal{L}}{\partial F^{(a}{}_c } F_{b)c}=0\,.
\label{eq:gravdual}
\end{equation}
Now, let us explore how this equation transforms under an infinitesimal duality rotation of angle $\alpha$. Direct computation shows:
\begin{equation}
\delta_\alpha \left (  \lambda \frac{\partial \mathcal{L}}{\partial \psi} \chi_{\rm GB}- \mathcal{L} \right)=\alpha \left ( -\lambda  \chi_{\rm GB}+2  p \right)\,, \quad \delta_\alpha \left ( \frac{\partial \mathcal{L}}{\partial F^{(a}{}_c } F_{b)c} \right)=\alpha g_{ab} \left (-p+\frac{\lambda \chi_{\rm GB}}{2} \right)\,.
\end{equation}
As a consequence, we clearly observe that \eqref{eq:gravdual} is invariant under duality rotations. In any event, we also note that the \eqref{eq:gravdual} will include higher derivatives of the metric and the gauge field.

\end{enumerate}

\section{Solving \eqref{eq:spdual} and \eqref{eq:hdifcond} perturbatively}\label{appC}
Let $\mathcal{L}(s,p,\psi)$ be a NED satisfying \eqref{eq:spdual} with nonzero $\psi$. Assume it is a homogeneous function of degree one in the variables $s,p$ and $\psi$. As a result:
\begin{equation}
\label{eq:hominv}
s \frac{\partial \mathcal{L}}{\partial s}+p \frac{\partial \mathcal{L}}{\partial p}+\frac{\partial \mathcal{L}}{\partial \psi}\psi=\mathcal{L}\,,
\end{equation}
Write the Lagrangian as a power series in $\psi$:
\begin{equation}
\mathcal{L}(s,p,\psi)=\sum_{n=0}^\infty \mathcal{L}_n(s,p)  \psi^{n}\,,
\label{eq:serpotlag}
\end{equation}
By substituting into the condition \eqref{eq:hominv}, we find:
\begin{equation}
s \frac{\partial \mathcal{L}_n}{\partial s}+p \frac{\partial \mathcal{L}_n}{\partial p}=(1-n) \mathcal{L}_n\,, \quad n \geq 0\,.
\label{eq:serpotconf}
\end{equation}
On the other hand, if we impose \eqref{eq:serpotlag} in the duality-invariance condition \eqref{eq:spdual} one gets:
\begin{equation}
\sum_{m=0}^n \left ( p \frac{\partial \mathcal{L}_m}{\partial s} \frac{\partial \mathcal{L}_{n-m}}{\partial s} -2s \frac{\partial \mathcal{L}_m}{\partial s} \frac{\partial \mathcal{L}_{n-m}}{\partial p}-p\frac{\partial \mathcal{L}_m}{\partial p} \frac{\partial \mathcal{L}_{n-m}}{\partial p} \right) =p \,\delta_{n,0}-\delta_{n,1}\,, \quad n \geq 0\,.
\label{eq:serpotdual}
\end{equation}
From these derivations, one can follow the algorithm below to construct novel families of duality-invariant theories of electrodynamics which are homogeneous functions of degree one:
\begin{enumerate}[leftmargin=*]
\item Start from a given seed duality-invariant theory of electrodynamics $\mathcal{L}_0$ which is a homogeneous function of degree one.
\item Take first $k=1$. Solve for $\frac{\partial \mathcal{L}_k}{\partial p}$ (or analogously, for $\frac{\partial \mathcal{L}_k}{\partial s}$) in Eq.\ \eqref{eq:serpotconf}. 
\item Substitute this result in \eqref{eq:serpotdual}. One gets a first-order differential equation for $\mathcal{L}_k$ in which only derivatives with respect to $s$ (or $p$, had we solved in step 2 for $\frac{\partial \mathcal{L}_k}{\partial s}$) appear. 
\item Solve such a differential equation and get $\mathcal{L}_k$. One gets an arbitrary function of $p$ (respectively, of $s$) to be fixed by demanding consistency with \eqref{eq:serpotconf}. On doing this, an arbitrary integration constant will remain unspecified. 
\item Repeat steps 2, 3 and 4 for $k \geq 2$. 
\end{enumerate}
The most general duality-invariant theory at a given perturbative order in $\psi$ satisfying the homogeneity condition \eqref{eq:hominv} will be obtained by taking ModMax theory as a seed, $\mathcal{L}_0=\mathcal{L}_{\rm MM}^\gamma$. Up to sixth order in $\psi$, we find:
\begin{align}
\mathcal{L}_1&=-\frac{1}{2} \arctan \left ( \frac{\mathcal{L}_{\rm MM}^\gamma }{p} \right)+\kappa_1\,, \\
\mathcal{L}_2&=-\frac{\mathcal{L}_{\rm MM}^\gamma }{8(\left ( \mathcal{L}_{\rm MM}^\gamma \right)^2+p^2)}+\frac{\kappa_2}{\sqrt{\left ( \mathcal{L}_{\rm MM}^\gamma \right)^2+p^2}}\,, \\
\mathcal{L}_3&=\frac{-\mathcal{L}_{\rm MM}^\gamma p+8\kappa_2 p \sqrt{\left ( \mathcal{L}_{\rm MM}^\gamma \right)^2+p^2}+16(p^2+\left ( \mathcal{L}_{\rm MM}^\gamma \right)^2) \kappa_3}{16(\left ( \mathcal{L}_{\rm MM}^\gamma \right)^2+p^2)^2}\,, \\ \nonumber
\mathcal{L}_4&=\frac{(5+192 \kappa_2^2)\left ( \mathcal{L}_{\rm MM}^\gamma \right)^3+384 \kappa_3 p(\left ( \mathcal{L}_{\rm MM}^\gamma \right)^2+p^2)+3 p^2 \left (-5\mathcal{L}_{\rm MM}^\gamma +64 \kappa_2^2 \mathcal{L}_{\rm MM}^\gamma +64 \kappa_2 \sqrt{\left ( \mathcal{L}_{\rm MM}^\gamma \right)^2+p^2} \right)}{384(\left ( \mathcal{L}_{\rm MM}^\gamma \right)^2+p^2)^3}\\& +\frac{\kappa_4}{(\left ( \mathcal{L}_{\rm MM}^\gamma \right)^2+p^2)^{3/2}}\,, \\
\nonumber
\mathcal{L}_5&=\frac{(320\kappa_2^2-7)\mathcal{L}_{\rm MM}^\gamma p^3+(320\kappa_2^2+7)p \left ( \mathcal{L}_{\rm MM}^\gamma \right)^3+320 \kappa_3 p^2(\left ( \mathcal{L}_{\rm MM}^\gamma \right)^2+p^2)}{256(\left ( \mathcal{L}_{\rm MM}^\gamma \right)^2+p^2)^4}\\&+\frac{12 \kappa_4 p (\left ( \mathcal{L}_{\rm MM}^\gamma \right)^2+p^2)+\kappa_2(4p^3+16 \kappa_3 \mathcal{L}_{\rm MM}^\gamma  p^2-p \left ( \mathcal{L}_{\rm MM}^\gamma \right)^2+16 \kappa_3 \left ( \mathcal{L}_{\rm MM}^\gamma \right)^3)}{8(\left ( \mathcal{L}_{\rm MM}^\gamma \right)^2+p^2)^{7/2}}+\frac{\kappa_5}{(\left ( \mathcal{L}_{\rm MM}^\gamma \right)^2+p^2)^{2}}\,,\\
\nonumber
\mathcal{L}_6&=\frac{p \left(-8 \kappa_2^3 p (\left ( \mathcal{L}_{\rm MM}^\gamma \right)^2+p^2)+\kappa_2 \left(48 \kappa_3 p^2 \mathcal{L}_{\rm MM}^\gamma +48 \kappa_3 \left ( \mathcal{L}_{\rm MM}^\gamma \right)^3+4 p^3-3 p \left ( \mathcal{L}_{\rm MM}^\gamma \right)^2\right)+18 \kappa_4 p \left(p^2+\left ( \mathcal{L}_{\rm MM}^\gamma \right)^2\right)\right)}{8 \left(p^2+\left ( \mathcal{L}_{\rm MM}^\gamma \right)^2\right)^{9/2}}\\\nonumber &+\frac{5 p^4 \mathcal{L}_{\rm MM}^\gamma  \left(2752 \kappa_2^2+3072 \kappa_2 \kappa_4-21\right)+10 p^2 \left ( \mathcal{L}_{\rm MM}^\gamma \right)^3 \left(1408 \kappa_2^2+3072 \kappa_2 \kappa_4+21\right)}{5120\left(p^2+\left ( \mathcal{L}_{\rm MM}^\gamma \right)^2\right)^5}\\\nonumber &+\frac{\left ( \mathcal{L}_{\rm MM}^\gamma \right)^5 \left(320 \kappa_2^2+15360 \kappa_2 \kappa_4-21\right)+10240 \kappa_3^2 \mathcal{L}_{\rm MM}^\gamma  \left(p^2+\left ( \mathcal{L}_{\rm MM}^\gamma \right)^2\right)^2+1280 \kappa_3 \left(6 p^5+5 p^3 \left ( \mathcal{L}_{\rm MM}^\gamma \right)^2-p \left ( \mathcal{L}_{\rm MM}^\gamma \right)^4\right)}{5120\left(p^2+\left ( \mathcal{L}_{\rm MM}^\gamma \right)^2\right)^5}  \\& +\frac{2 \kappa_5 p}{\left(p^2+\left ( \mathcal{L}_{\rm MM}^\gamma \right)^2\right)^3}+\frac{\kappa_6}{(\left ( \mathcal{L}_{\rm MM}^\gamma \right)^2+p^2)^{5/2}}\,,
\end{align}
where $\{\kappa_1,\dots, \kappa_6\}$ are arbitrary real constants. We observe that each order $n$ introduces a new integration constant. Continuing with this procedure for arbitrary $n$, one would end up with a tower of arbitrary constants $\{\kappa_n\}_{n=1}^\infty$, which signals the fact that there is a one-variable worth of Lagrangians satisfying simultaneously Eq.\ \eqref{eq:spdual} and Eq.\ \eqref{eq:hominv}. In another vein, let us highlight that the theory $\mathcal{L}_{\rm \psi M}$ presented in \eqref{eq:ame} is obtained by taking to $n \rightarrow \infty$ the perturbative algorithm we have just described and setting all integration constants $\kappa_n=0$.


\twocolumngrid

\bibliographystyle{JHEP-2}
\bibliography{Gravities}
\noindent 


\end{document}